\documentclass[a4paper,12pt]{amsart}
\usepackage[utf8]{inputenc}
\usepackage{fullpage}
\usepackage{hyperref}
\usepackage{enumerate}
\usepackage[usenames,dvipsnames]{xcolor}
\usepackage{amsmath,amsthm,amsfonts,amssymb}
\usepackage{mathtools,mathabx}
\usepackage{graphicx}

\newcommand{\bd}[1]{\mathrm{bd}\left(#1\right)}
\newcommand{\norm}[1]{\left\|#1\right\|}
\newcommand{\cvp}[1][]{\ensuremath{\text{CVP}_{#1}}}
\newcommand{\acvp}[2][(1+\epsilon)]{\ensuremath{#1}-\cvp[#2]}
\newcommand{\svp}[1][]{\ensuremath{\text{SVP}_{#1}}}
\newcommand{\twocov}[1][(2,\epsilon)]{\ensuremath{#1}-covering }
\newcommand{\twocovs}{$(2,\epsilon)$-coverings }
\renewcommand{\epsilon}{\varepsilon}
\newcommand{\lattenum}{{\tt Lattice-Enumerator}}
\newcommand{\lattspars}{{\tt Lattice-Sparsifier}}
\newcommand{\di}{\mathrm{d}}
\newcommand{\ceil}[1]{\left\lceil #1 \right\rceil}
\newcommand{\floor}[1]{\left\lfloor #1 \right\rfloor}

\makeatletter
\@namedef{subjclassname@2020}{%
  \textup{2020} Mathematics Subject Classification}
\makeatother

\def\R{\mathbb{R}}
\def\Q{\mathbb{Q}}

\def\Z{\mathbb{Z}}

\def\Ren{\mathbb{R}^n}
\def\st{\;:\;}

\DeclareMathOperator{\vol}{vol}
\DeclareMathOperator{\conv}{conv}

\DeclareMathOperator{\poly}{poly}

\author{M\'arton Nasz\'odi}
 \address{MN: Alfr\'ed R\'enyi Inst. of 
 Math.; MTA-ELTE Lend\"ulet Combinatorial Geometry Research Group;
 Dept. of Geometry, Lor\'and E\"otv\"os University, Budapest  
}
\email{marton.naszodi@math.elte.hu}

\author{Moritz Venzin}
 \address{MV: Institute for Mathematics, \'Ecole Polytechnique 
 Fédérale de Lausanne, Lausanne, Switzerland
}
\email{moritz.venzin@epfl.ch}

\keywords{Closest Vector Problem, Modulus of 
smoothness, Lattice sparsification, Convex body in d-dimensional space, 
Approximation}
\subjclass[2020]{90C10, 52C07, 68W25, 68Q25, 68U05}

\newtheorem{theorem}{Theorem}[section]
\newtheorem{lemma}[theorem]{Lemma}
\newtheorem{proposition}[theorem]{Proposition}

\newtheorem{corollary}[theorem]{Corollary}

\theoremstyle{definition}
\newtheorem{definition}[theorem]{Definition}

\newtheorem{remark}[theorem]{Remark}

\title{Covering convex bodies and the closest vector problem}

\begin{document}

\begin{abstract}
We present algorithms for the $(1+\epsilon)$-approximate version of the closest 
vector problem for certain norms.  The currently fastest algorithm (Dadush and 
Kun 2016) for general norms in dimension $n$ has running time of $2^{O(n)} 
(1/\epsilon)^n$. We 
improve this substantially in the following two cases.

First, for $\ell_p$-norms with $p>2$ (resp. $p \in [1,2]$) fixed,  we present 
an 
algorithm with a  running time of $2^{O(n)} (1+1/\epsilon)^{n/2}$ 
(resp. $2^{O(n)} (1+1/\epsilon)^{n/p}$). This result is based on a geometric 
covering problem, that was introduced in the context of CVP by Eisenbrand et 
al.: How many convex bodies are needed to cover the ball of the norm such that, 
if scaled by factor 2 around their centroids, each one is contained in the 
$(1+\epsilon)$-scaled homothet of the norm ball? We provide upper 
bounds for this \emph{$(2,\varepsilon)$-covering number} by exploiting the  
\emph{modulus of smoothness} of the 
$\ell_p$-balls. Applying a covering scheme, we can boost any 
$2$-approximation algorithm for CVP to a $(1+\epsilon)$-approximation algorithm 
with the improved run time, either using a straightforward sampling routine or 
using the deterministic algorithm of Dadush for the construction of an epsilon 
net.

Second, we consider polyhedral and zonotopal norms.  For centrally 
symmetric polytopes (resp. zonotopes) in ${\mathbb R}^n$ with $O(n)$ facets 
(resp. generated by 
$O(n)$ line segments), we provide a deterministic  
$O(\log_2(2+1/\epsilon))^{O(n)}$ time algorithm. This generalizes the result of 
Eisenbrand et al. which applies to the $\ell_\infty$-norm.

Finally, we establish a connection between the \emph{modulus of smoothness} and 
\emph{lattice sparsification}. As a consequence, using the enumeration and 
sparsification tools developped by Dadush, Kun, Peikert and Vempala, we present 
a simple alternative to the boosting procedure with the same time and space 
requirement for $\ell_p$ norms. This connection might be of independent 
interest.
\end{abstract}
\maketitle

\section{Introduction}
The \emph{closest vector problem} (CVP) is an important algorithmic problem in 
the 
geometry of numbers. Given a rational lattice $\Lambda(A) = \{Ax \st x \in 
\Z^n\}$, with $ A \in \Q^{n \times n}$ and a target vector $t \in \Q^n$, the 
task 
is to find a close vector in $\mathcal{L}$ to $t$ with respect to a given 
norm. Specifically, given some norm $\|\cdot\|_K$, a $(1+\epsilon)$-approximation to the closest vector problem, \acvp{K}, is to find a lattice vector whose distance to the target vector 
is at most $(1+\epsilon)$ times the minimal distance of the target to the 
lattice. Whenever $K$ is 
the unit ball of the space $\ell_p^n$ for some $1\leq p\leq \infty$, we denote the problem by \acvp{p}. The closely related \emph{shortest vector problem} (SVP) asks for the shortest non-zero 
lattice vector in a given lattice.
It was shown that CVP is NP-hard for 
any $\ell_p$ norm \cite{vanEmdeBoas81} and even NP-hard to approximate up to 
almost polynomial factors, \cite{Arora:1995:PCP:220989}, 
\cite{DBLP:journals/combinatorica/DinurKRS03}. 

The first algorithm to solve integer programming and, in particular, exact
\cvp[\infty] was given by Lenstra \cite{DBLP:journals/mor/Lenstra83} 
with a running time of $2^{O(n^2)}$. His algorithm connects the two fields of 
geometry of numbers and integer programming. Kannan 
\cite{DBLP:journals/mor/Kannan87} presented an algorithm for exact \cvp[] (and \svp[])
with a running time of $n^{O(n)}$ and polynomial space. Subsequent works 
improve on the constant in the exponent but improving the running time of 
$n^{O(n)}$ to single exponential in $n$ remained an open problem.  After 
Kannan's result, it took almost 15 years until Ajtai, Kumar and Sivakumar 
presented a randomized algorithm for $\text{SVP}_2$ with time and space 
$2^{O(n)}$ and 
\acvp{2} with time and space $2^{(1+1/\epsilon)n}$, 
\cite{DBLP:conf/stoc/AjtaiKS01}, 
\cite{DBLP:conf/coco/AjtaiKS02}. Subsequently, Blömer and Naewe 
\cite{DBLP:journals/tcs/BlomerN09} extended the randomized sieving algorithm of 
Ajtai et al. to solve \acvp{p} for all $p$ in time $O(1+1/\epsilon)^{2n}$ and 
space $O(1+1/\epsilon)^n$, see also \cite{svp_l_infty} and \cite{svp_mukh_l_p}. 
For $p = \infty$, Eisenbrand, Hähnle and Niemeier 
\cite{DBLP:conf/compgeom/EisenbrandHN11} then boosted the algorithm of Blömer 
and Naewe 
by showing that $2^{O(n)}\log(2+1/\epsilon)^n$ calls to a 
\acvp[2]{\infty} solver suffice to solve 
\acvp{\infty} implying a running time of 
$O(\log(2+1/\epsilon))^n$ and space requirement $2^{O(n)}$. Dadush \cite{DBLP:conf/latin/Dadush12} extended the Ajtai--Kumar--Sivakumar 
sieve to solve \acvp{} in any norm with a running time of $O(1+1/\epsilon)^{2n}$ 
and space $O(1+1/\epsilon)^n$. 
The first single exponential deterministic and exact solver for \cvp[2] was presented by Micciancio and Voulgaris 
\cite{DBLP:conf/stoc/MicciancioV10}. Their algorithm needs to store the up to 
$2(2^n-1)$ facets 
of the Voronoi cell of the lattice. Recently in 
\cite{DBLP:conf/ipco/HunkenschroderR19}, Hunkenschröder, Reuland and Schymura 
show that this can be avoided and do a first step towards a polynomial space 
algorithm for \cvp[2]. The currently fastest algorithms for exact \cvp[2] and 
$\text{SVP}_2$ use discrete
Gaussian sampling and need time and space $2^{n+o(n)}$, see \cite{gaussian_sampling}, 
\cite{DBLP:conf/soda/AggarwalS18}. Despite this progress for the $\ell_2$ norm, 
for general norms, only the randomized sieving approach seemed available to 
solve CVP. Using the elegant idea of lattice sparsification, Dadush and Kun 
\cite{DBLP:journals/toc/DadushK16} presented a deterministic algorithm 
solving \acvp{} for any norm in time 
$2^{O(n)}(1+1/\epsilon)^n$ and with space requirement $2^n \poly(n)$ - reducing 
the dependence on $(1/\epsilon)$ in the running time and removing the 
dependence on $(1/\epsilon)$ in the space requirement altogether compared with 
earlier randomized sieving approaches.

\subsection*{Our contribution}
In order to devise more efficient algorithms for \cvp[K] (and, in particular 
\cvp[p]), we study the problem of how many arbitrarily chosen convex bodies are needed to cover 
some given convex body $K$, such that when scaled around their respective centroids 
by a factor $2$, each one is contained in $(1+\epsilon)K$. We refer to such a 
covering as a \emph{\twocov for $K$}, and the smallest size of such a covering 
as the \emph{\twocov number of $K$}.

A key quantity, well studied in the theory of Banach spaces, is the 
\emph{modulus of smoothness} of a convex body $K$, which expresses how well the 
boundary of $K$ is approximated locally by support hyperplanes, see 
Definition~\ref{def:modsmoothness}.

In this paper the \emph{big oh notation}, $O(.)$, stands for a universal 
multiplicative constant independent of every other quantity. In particular, we 
make the dependence 
on $\varepsilon$ and $n$ explicit.

\begin{enumerate}
	\item  By a standard argument, we show that for any centrally symmetric 
convex body, a \twocov is always possible using 
$2^{O(n)}(1+\frac{1}{\epsilon})^n$ convex bodies. Then, in 
Theorem~\ref{thm:ballneedsmany}, we establish a \textbf{lower bound} of 
$2^{-O(n)}(1+\frac{1}{\epsilon})^{n/2}$ for the Euclidean unit ball. 
	\item For centrally symmetric polytopes (resp. zonotopes) with \textbf{$m$ 
facets} 
(resp. $m$ generating line segments), we provide an explicit \twocov using at 
most $O(\log(2+\frac{1}{\epsilon}))^m$ convex bodies, see 
Propositions~\ref{covering:zonotope} and \ref{covering:polytope}.
These are relatively straightforward generalizations of the method of 
\cite{DBLP:conf/compgeom/EisenbrandHN11} where the cube is considered.
	\item Our first main result is Theorem~\ref{cov:smooth}, where it is 
shown 
that a bound on the \textbf{modulus of smoothness} of $K$ yields a 
\textbf{bound on its
\twocov number}. More specifically, if 
$K$ has modulus of smoothness bounded above by $C\tau^q$, then we 
find a \twocov of $K$ using $C^{O(n)}(1+\frac{1}{\epsilon})^{n/q}$ convex 
bodies. In particular, we obtain a \twocov for $\ell_p$ balls 
using $2^{O(n)}(1+\frac{1}{\epsilon})^{n/2}$ for $p \geq 2$ and 
$2^{O(n)}(1+\frac{1}{\epsilon})^{n/p}$ for $p \in [1, 2]$, matching the lower 
bound (Theorem~\ref{thm:ballneedsmany}) for the Euclidean unit ball.
	\item 
Our second main result is Theorem~\ref{2-epsilon:det-solver}, which shows how a 
good algorithmic bound on the \twocov number yields an \textbf{efficient 
\acvp{} algorithm}.
In particular, for norms induced by centrally symmetric polytopes (resp. 
zonotopes) 
with $m$ facets (resp. generating line segments), the above explicit 
\twocov boosts any \acvp[2]{} solver for general norms 
to yield a deterministic \acvp{} algorithm. This yields 
an algorithm with running time $O(\log(2+\frac{1}{\epsilon}))^m$ and $2^n 
\poly(n)$ space, see Corollary~\ref{cor:algforpolytopeszonotopes}.

\item For a centrally symmetric convex body $K$ with a certain modulus 
of smoothness, to avoid the space requirement to depend on the number of convex 
bodies in the \twocov of $K$, we show how to \textbf{generate a local 
\twocov on the fly}. This yields a simple, randomized 
\acvp{p} algorithm for $1\leq p\leq\infty$ with a 
running time of 
$O(1+\frac{1}{\epsilon})^{n/2}$ for $p \geq 2$, and 
$2^{O(n)}(1+\frac{1}{\epsilon})^{n/p}$ for $p \in [1,2]$, using $2^{n} \poly(n)$ 
space. Alternatively, we may use an algorithm of Dadush \cite{Da2013} to 
explicitly enumerate the covering using polynomial space only, derandomizing 
the 
algorithm. This is our third main result, see 
Theorem~\ref{thm:algmodsmoothness}. 

Compared to earlier results 
in the literature, for instance \cite{DBLP:journals/tcs/BlomerN09}, 
\cite{DBLP:journals/toc/DadushK16}, we improve on the previous best running 
times of $O(1+\frac{1}{\epsilon})^n$ for $\ell_p$ norms. 

Furthermore, our approach immediately generalizes to non-symmetric norms and we 
obtain a simple CVP solver for $\gamma$-symmetric norms with running time 
$(1+\frac{1}{\gamma  \epsilon})^n$ and space requirement $2^{O(n)}$ 
based on the Ajtai--Kumar--Sivakumar sieve, see Remark~\ref{rem:nonsymm}. This 
almost matches the performance of Dadush and Kun's algorithm.

\item Finally, we establish a connection between \textbf{lattice 
sparsification} and the 
\textbf{modulus of smoothness}, see Lemma~\ref{lemma:observation}. While the 
boosting 
approach described in Sections~\ref{sec:modsmoothness} and 
\ref{sec:cvpalg} is conceptually very simple and general, and it does not 
require any knowledge about the approximate \cvp\, solver used, the proofs are 
quite technical. 
We will show that we can tweak the algorithm described by 
Dadush and Kun in \cite{DBLP:journals/toc/DadushK16} using a simple observation 
based on the modulus of smoothness in order to obtain the same improved running 
time for \cvp\, for norms with a certain modulus of smoothness, in particular 
\cvp[p]. With this new approach, we restrict ourselves to using lattice 
sparsification and enumeration and we lose the possibility to use an arbitrary 
constant 
approximation \cvp-solver. Considering the low space dependency of lattice 
sparsification and 
enumeration among all known (single 
exponential) approximate \cvp\, solvers and the simplicity of our approach, 
this might not be a big loss.
\end{enumerate}

It should be noted here that a seemingly similar (with respect to 
$\varepsilon$) bound on the \twocov number follows from recent work of Arya, 
Fonseca and Mount \cite{AFM17} (see also \cite{AM18}). Using Macbeath regions, 
they approximate \emph{any} convex body with a polytope with at most 
$n^{O(n)}\varepsilon^{-(n-1)/2}$ faces of all dimensions in total, 
provided that $\varepsilon \ll n^{-n}$. It is then
straightforward to show that this can be turned into a \twocov using roughly
$n^{O(n)}\varepsilon^{-(n-1)/2}$ convex bodies. 
Unfortunately, for the purpose of designing 
approximation algorithms for lattice problems, this is of little use, as already 
the $n^{O(n)}$ factor is prohibitively high considering that the exact solver 
of Kannan runs in $n^{O(n)}$ time. Moreover, any approximation based on 
Macbeath regions requires $\varepsilon \ll n^{-n}$, which is too strong a 
restriction for integer programming related applications. Nonetheless, their 
result shows that for $\varepsilon$ sufficiently small, any convex body admits a 
\twocov using $O(1+1/\varepsilon)^{n/2}$ convex bodies and raises the question 
whether the restriction on $\varepsilon$ can be removed in general.  As 
mentioned above, in the present work, the dimension $n$ 
is not considered constant, and dependence on it is made explicit everywhere.

The structure of the paper is the following. In Section~\ref{sec:twocov}, we 
list basic facts about \twocovs and prove upper bounds on the \twocov number 
of symmetric polytopes and zonotopes (Propositions~\ref{covering:zonotope}
and ~\ref{covering:polytope}). In Theorem~\ref{thm:ballneedsmany}, a lower 
bound 
on the covering number of the Euclidean ball is presented. In 
Section~\ref{sec:modsmoothness}, it is shown how a bound on the modulus of 
smoothness yields a bound on the \twocov number. In 
Section~\ref{sec:cvpalg}, we apply our covering bounds to obtain efficient 
algorithms for \acvp{}. Finally, Section~\ref{sec:sparsifier} contains 
Theorem~\ref{thm:sparsificationmodulus}, which presents another \acvp{} solver 
for bodies with a well bounded modulus of convexity, based on efficient lattice 
sparsification and lattice enumeration algorithms.

The scalar product of two vectors $x=(x_1,\ldots,x_n)$ and $y=(y_1,\ldots,y_n)$ 
in $\Ren$ is denoted by $\langle x,y\rangle=x_1y_1+\ldots+x_ny_n$. For a 
positive integer $k$, we use the notation $[k]=\{1,\ldots,k\}$.

\section{\texorpdfstring{$(2,\epsilon)$}{(2,epsilon)}-coverings}
\label{sec:twocov}
We denote the \emph{homothetic copy} of a convex body $Q$ by factor 
$\lambda\in\R$ with respect to its \emph{centroid} (also called, center of 
mass) $c(Q)$ by $\lambda\odot Q=\lambda(Q-c(Q))+c(Q)$.

The following notion is central to our study.

\begin{definition}[$(2,\epsilon)$-covering]\label{defn:2epiloncovering}
For a convex body $K \subseteq \R^n$, a sequence of convex bodies 
$\{Q_i\}_{i=1}^N$ is a \twocov if 
\[K \subseteq \bigcup_{i=1}^N Q_i \subseteq \bigcup_{i=1}^N 2\odot Q_i 
\subseteq (1+\epsilon)K.\]
\end{definition}
We note that we have fixed the factor $2$ for concreteness, we could replace $2$ by any other constant. For this reason we will assume  $\epsilon \in (0,1)$.

The following three lemmas follow directly from standard packing arguments, we 
include a proof in the Appendix \ref{sec:appendix}.
 
\begin{lemma}\label{lemma:symmetric}
	Any origin symmetric convex body $K \subseteq \R^n$ admits a \twocov by at 
most $(\frac{5}{\epsilon})^n$ homothetic copies of $K$.
\end{lemma}

We also note that it is sufficient to consider coverings by centrally 
symmetric convex bodies only.
\begin{lemma}\label{lem:symmetriccovering}
Let $K$ be a convex body in $\Ren$ that admits a \twocov 
consisting of $N$ convex bodies.
Then, $K$ admits a \twocov consisting of $10^n N$ centrally 
symmetric convex bodies.
\end{lemma}

\begin{lemma}\label{lemma:symmetrizer}
Any convex body $K \subseteq \R^n$ with $0$ as its centroid has a \twocov by at 
most $N = 
(\frac{10}{\epsilon})^n$ translated copies of $\frac{\epsilon}{2} (K\cap -K)$.
\end{lemma}

In the particular case of the cube, in 
\cite{DBLP:conf/compgeom/EisenbrandHN11}, Eisenbrand et al. 
found a \twocov that requires $(1+2\log_2(1+1/\epsilon))^n$  
parallelepipeds. The following two propositions show that their method 
generally works for any zonotope or any centrally symmetric polytope.

A \emph{zonotope} is the Minkowski sum of finitely many line segments, 
$\mathcal{Z} = \{ \sum_{i = 1}^m \lambda_i b_i \st \lambda_i \in [-1,1], \, 
\, 1 \leq i \leq m\} = \sum_{i=1}^m [-b_i, b_i]$. We 
refer to the $b_i$ as the \emph{generators} of $\mathcal{Z}$. If $m=n$ and $b_i 
= e_i\; (i=1,\ldots,n)$, then 
this zonotope is the unit cube. A zonotope with $m$ generators can have up to 
$2\binom{m}{n-1}$ facets; when no $n$ of the generators are linearly dependent, 
this bound is attained, as is not difficult to see. 

In the following two Propositions, we give upper bounds for the \twocov of 
zonotopes with a bounded number of generators and for polytopes with a bounded 
number of facets.
We include these proof in the Appendix \ref{sec:appendix}.

\begin{proposition}[\twocov of a zonotope by smaller 
zonotopes]\label{covering:zonotope}
Let $\mathcal{Z} = \{\sum_{i=1}^m \lambda_i b_i \st \lambda_i \in [-1,1] 
\text{ and } i \in [m]\}$ be a zonotope with $m$ generators, $b_1, \ldots, b_m \in \R^n$. For any $\epsilon 
>0$, there exists a \twocov of $\mathcal{Z}$ using 
$(1+2\log_2(1+1/\epsilon))^m$ zonotopes.
\end{proposition}

\begin{proposition}[\twocov centrally symmetric polytopes with 
few facets]
\label{covering:polytope}
Let $P = \{x\in \R^n \st |a_i^T x|\leq b_i \, , i \in [m]\}$ be a 
origin symmetric polytope.
There is a \twocov of $P$ using at most $2^m 
(\log_{4/3}(1/\epsilon)+1)^m$ centrally symmetric convex bodies.  
\end{proposition}

Finally, we prove a lower bound on the \twocov number of the 
Euclidean unit ball $B_2^n$ which, by Corollary~\ref{cor:twocovlp}, is sharp, 
up to a logarithmic factor.

\begin{theorem}\label{thm:ballneedsmany}
For any $\varepsilon\in(0,1/2)$, any \twocov of the Euclidean unit ball $B_2^n$ 
consists of at 
least \allowbreak $2^{-O(n)}(1/\epsilon)^{(n-1)/2}$ convex bodies.
\end{theorem}
\begin{proof}
Let $\{Q_i\}_{i=1}^N$ be a \twocov of $B_2^n$ with respective 
centroids $c_i$. Let $p \in \mathbb{S}^{n-1}$ and let $c$ be the centroid of a 
$Q_i$ such that $p \in Q_i$. 
First, we show that $\langle p, c \rangle \geq 1-\epsilon$, that is, $Q_i$ is 
contained in a small solid cap. Suppose by contradiction that $\langle p, c 
\rangle < 
1-\epsilon$. By the definition of a \twocov we need that $\|p + 
(p-c)\| \leq 1+\epsilon$. This implies $\langle p, p + (p-c)\rangle \leq 
1+\epsilon$ and we obtain the following contradiction:
\[\langle p, p+(p-c)\rangle = 2\langle p, p\rangle + \langle p, -c\rangle > 2 + 
\epsilon -1 = 1 + \epsilon.\]
Also by the definition of a \twocov, we need $\|c\| \leq 
1+\epsilon$. Thus, we can show $\|p-c\|$ is small:
\begin{alignat*}{1}
\langle p-c, p-c\rangle &= \langle p, p\rangle + \langle c, c \rangle + 
2\langle p, -c\rangle\\
&\leq 1 + (1+\epsilon)^2 + 2(\epsilon -1)\\
&\leq 5\epsilon.
\end{alignat*}
Thus, for every $Q_i$, $Q_i \cap \mathbb{S}^{n-1}$ is contained in a cap of 
radius $\sqrt{5\epsilon}$. 
Denoting by $\sigma(\cdot)$ the uniform probability measure on the sphere, this 
means that for any convex body $Q_i$ in the \twocov, $\sigma(Q_i) \leq 
2^{O(n)}\epsilon^{(n-1)/2}$ (cf. \cite[Lemma~3.1]{BW03}). Since a \twocov of 
$B_2^n$ needs to cover 
all of $\mathbb{S}^{n-1}$, we obtain the desired lower bound on $N$.
\end{proof}

\section{\texorpdfstring{$(2,\epsilon)$}{(2,epsilon)}-coverings via modulus of 
smoothness}\label{sec:modsmoothness}

For a convex body $K$, we will consider its \emph{gauge function} 
$\norm{\cdot}_K$, 
defined by $\norm{x}_K = \inf\{s \st x \in sK\}$. If $K$ is origin symmetric, 
then $\norm{\cdot}_K$ defines a norm. 
\begin{definition}[Modulus of smoothness]\label{def:modsmoothness}
The \emph{modulus of smoothness} of an origin-symmetric convex body $K$, 
$\rho_K(\tau): (0,1) \rightarrow (0,1)$, is defined by 
\[\rho_K(\tau) = \frac{1}{2} \sup_{\norm{x}_K = \norm{y}_K = 1}(\norm{x+\tau 
y}_K + \norm{x-\tau y}_K 
-2).\]
\end{definition}
We remark first that any origin symmetric body $K$ has modulus of smoothness 
$\rho_K(\tau) \leq \tau$, this follows from the subadditivity of the norm.
The modulus of smoothness of $K$ measures how well $K$ can be locally 
approximated by hyperplanes: If $\|x\|_K = 1$ and $\|\tau y\|_K = \tau$ and 
both $x+y$ and $x-y$ lie on a support hyperplane of $K$ at $x$, then both 
$\|x+\tau y\|_K,\|x-\tau y\|_K \geq 1$, but we also have the upper bound of
\[\|x\pm \tau y\|_K \leq 1 + 2\rho_K(\tau).\]
If $\rho_K(\tau)$ can be bounded by a polynomial of degree higher than $1$, say 
$\tau^2$, then $x\pm \tau y$ are closer to the boundary of $K$ compared to what 
subadditivity, $\|x\pm \tau y\|_K \leq \|x\|_K + \|\tau y\|_K$, alone yields. 
Still assuming $\rho_K(\tau) \leq \tau^2$ and letting $\epsilon \in (0,1)$, 
this means that all points $y\in K$ with $\|x-y\| \leq \sqrt{\epsilon}$ are 
approximated up to an additive $\epsilon$ by the tangential hyperplane at $x$. 
This behaviour of some norms is exploited in the next theorem.
\begin{theorem}\label{cov:smooth}
Let $K \subseteq \R^n$ be an origin symmetric convex body, and 
$\epsilon\in(0,1)$. Assume 
that the modulus of smoothness of $K$ is bounded by 
\[\rho_K(\tau) \leq C\tau^q\]
with some constants $C,q>1$.
Then, there exists a \twocov of $K$ consisting of 
\[2^{O(n)}\log\left(1+1/\epsilon\right)\left(\frac{C}{\epsilon}\right)^{n/q
}+O(C)^{n/(q-1)}\]
centrally symmetric convex bodies. The encoding length of each such body is a polynomial in the encoding length of $K$.
\end{theorem}

\begin{figure}[ht]
 \includegraphics[width=.3\textwidth]{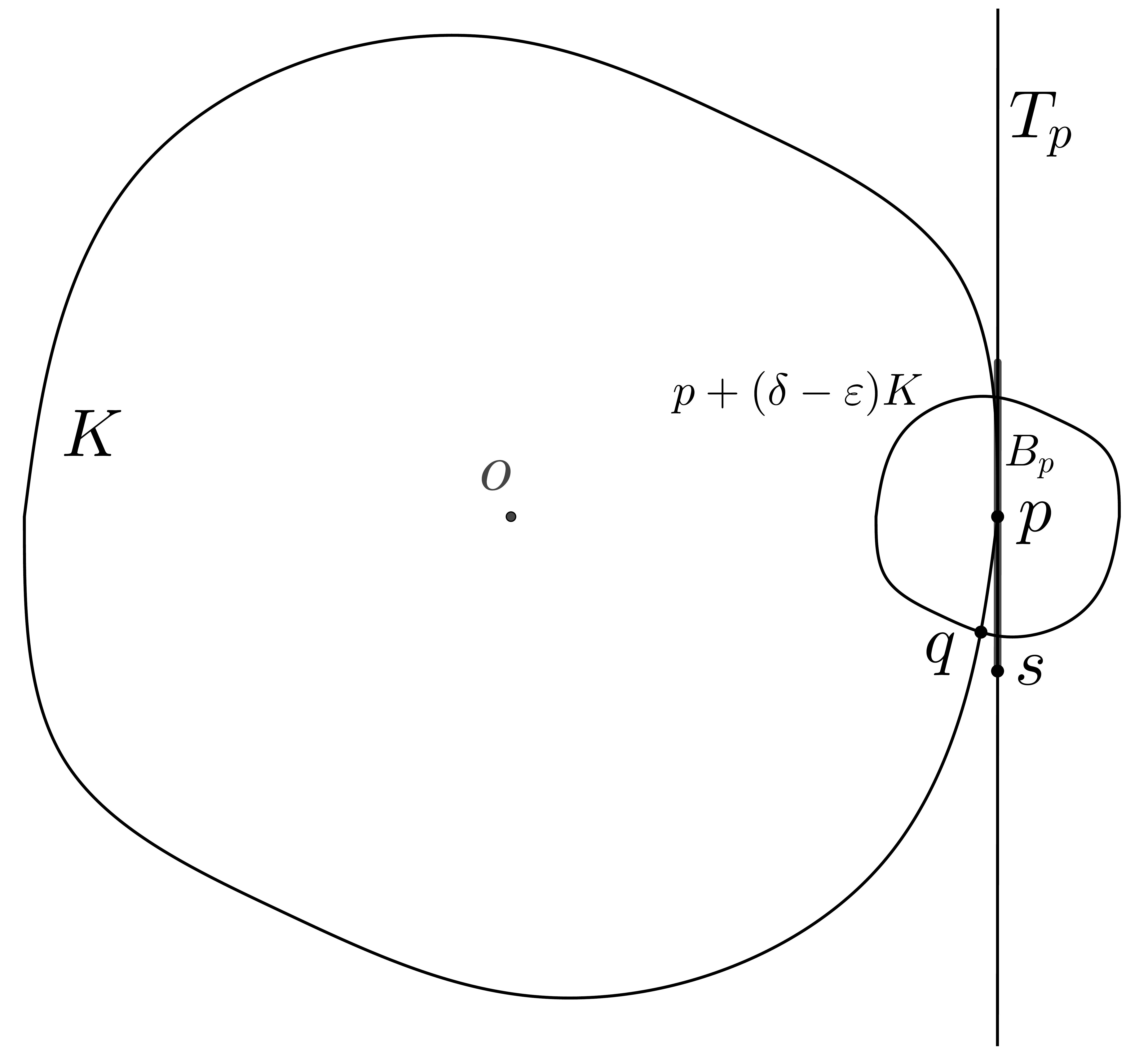}
\caption{Proof of \eqref{eq:capscoverbd}.}\label{fig:smoothtangent}
\end{figure}

\begin{proof}
Set $\delta = \frac{1}{4} \left(\frac{\epsilon}{C}\right)^{1/q}$. 
We may assume that $\epsilon \leq 
\left(\frac{1}{8C^{1/q}}\right)^{q/(q-1)}$, in which case $\delta-\epsilon 
\geq 
\delta/2$. Otherwise, we may apply Lemma~\ref{lemma:symmetric} and obtain 
a \twocov of $K$ consisting of $O(C)^{n/(q-1)}$ bodies. 
We denote $\norm{\cdot}_K$ by $\norm{\cdot}$.

We first describe a \twocov[(2,2\epsilon)] of $K$ only in the neighborhood 
of a point and then, using a packing argument, we extend this 
construction to obtain a \twocov[(2,2\epsilon)] for all of $K$.

Fix a point $p$ on the boundary of $K$ that is, $\norm{p} = 1$. Denote by $T_p$ 
a supporting hyperplane of $K$ at $p$. Let $B_p$ be the intersection of $T_p$ 
with $p + \delta K$, i.e. $B_p := T_p \cap \{x \st \norm{x-p} \leq \delta\}$. 

First, we show that 
\begin{equation}\label{eq:capscoverbd}
\bd{K} \cap \left(p+(\delta - \epsilon)K\right)\subseteq \conv(0, B_p).
\end{equation}
Indeed, let $q$ be a point in $\bd{K} \cap \left(p+(\delta - 
\epsilon)K\right)$, and let $L$ 
denote the two-dimensional linear plane spanned by $p, q$ and the origin $o$, 
see 
Figure~\ref{fig:smoothtangent}.
Clearly, $L\cap T_p$ is a line, and there are two points on this line at 
distance $\delta$ from $p$. Let $s$ denote the point of these two which is on 
the same side of the line $op$ as $q$. That is, $s$ is a point on the lateral 
surface of the cone $\conv(0, B_p)$. By the assumption on the modulus of 
smoothness of $K$, we have $s^\prime:=s/\norm{s}$ is at distance at most 
$\epsilon$ from $s$ (a detailed computation of a similar fact is given below in 
this proof). Thus,  
\begin{equation}\label{eq:sisfar}
 \norm{s^\prime-p}\geq\delta-\epsilon.
\end{equation}
Now, $L$ is a normed plane with unit circle $K\cap L$ and $p$ is a unit vector 
in $L$. It is a classical fact in the theory of normed planes 
\cite[Proposition~31]{martiniswanepoel} that as a point moves along the curve 
$K\cap L$ starting at $p$ and ending at $-p$, the distance (w.r.t. $\|\cdot\|_K$) of 
the moving point to $p$ is increasing. Thus, by \eqref{eq:sisfar}, the arc of 
$K\cap L$ between $p$ and $s^\prime$ contains $q$, which yields that 
$q$ is in the cone $\conv(0, B_p)$, proving \eqref{eq:capscoverbd}.

Next, instead of the cone $\conv(0, B_p)$, we will consider the cylinder 
\begin{alignat*}{1}
C_p = B_p + [0, -p].
\end{alignat*}
Clearly,  we have $\conv(0, B_p) \subseteq C_p$.

We may assume that $\epsilon$ is of the form 
$\epsilon=\left(2^{k}-1\right)^{-1}$, where $k$ is a positive integer.
For $i \in [k]$, consider the following slice of $C_p$:
\begin{alignat}{1}\label{mod.cylinders}
C_p(i) = \left(B_p + 
[-(2^i-1)\epsilon p, -(2^{i-1}-1)\epsilon p]\right).
\end{alignat}

Clearly, 
$2\odot C_p(i)\subseteq \widehat{C_p} := 2\odot B_p + [\epsilon p, 
-\frac{3}{2}p]$ and the centroid $c(C_p(i))$ is at 
$(1-(\frac{3}{2}2^{i-1}-1))\epsilon p$ for each $i\in[k]$.

We claim that $\widehat{C_p} \subseteq (1+2\epsilon)K$. 
Since $\delta\leq 1/4$ and $K=-K$, we have $2\odot B_p -\frac{3}{2}p\subseteq 
K$. Thus, it suffices to check that $2\odot B_p + \epsilon p\subseteq 
(1+2\epsilon)K$.

Let $x \in 2\odot B_p + \epsilon p$, i.e. $x = p + 2(z-p) + \epsilon p$ for 
some $z 
\in B_p$. We will show that $\|p + 2(z-p)\|\leq 1+2\epsilon$.
Since both $p$ and $z$ lie in $T_p$, then so do $p + 2(z-p)$ and $p + 2(p-z)$, 
and thus, we have $\|p + 2(z-p)\|, \|p + 2(p-z)\| \geq 1$.

$\|2(z-p)\| \leq 2\delta = \frac{1}{2}\left(\frac{\epsilon}{C}\right)^{1/q}$ 
and so by the 
assumption on the modulus of smoothness of $K$, we obtain
\[\|p + 2(z-p)\| \leq 2C\|2(z-p)\|^q + 1\leq 1 + \epsilon.\]
Thus, $\widehat{C_p} \subseteq (1+2\epsilon)K$, and hence,
\begin{equation*}
2\odot C_p(i) \subseteq (1+2\epsilon)K 
\end{equation*}
for each $i \in [k]$.

Since, by \eqref{eq:capscoverbd}, all points on the boundary of $K$ at distance 
at most $\delta 
- \epsilon$ from $p$ are covered by $C_p$, we see that all points $x$, such 
that $\|\frac{x}{\|x\|}-p\|\leq \delta - \epsilon$ are covered by one of the 
slices of $C_p$. Thus, in order to extend the above construction to a 
\twocov[(2,2\epsilon)] of $K$, we pick points 
$\{p_i\}_{i=1}^N$ on the boundary of $K$ such that $\bd K \subseteq 
\bigcup_{i=1}^N p_i + (\delta  - \epsilon)K$. By Lemma \ref{lemma:symmetric}, 
\[N = 
2^{O(n)}\left(\frac{1}{(\delta - \epsilon)}\right)^n = 
2^{O(n)}\left(\frac{C}{\epsilon}\right)^{n/q}
\]
such points suffice.

Thus, we obtain a \twocov[(2,2\epsilon)] for $K$ by constructing $C_{p_i}$ for 
each $i \in 
[N]$ and slicing each $C_{p_i}$ as in \eqref{mod.cylinders}. Finally, replacing 
$\epsilon$ by $\frac{\epsilon}{2}$, we indeed get a \twocov of $K$ using 
$2^{O(n)}(\frac{C}{\epsilon})^{n/q}\log\left(\frac{1}{\epsilon}\right)$ convex 
bodies, each described by a polynomial in the encoding length of $K$, see \cite{GroetschelLovaszSchrivjer88}.
\end{proof}

\begin{theorem}[Modulus of smoothness for $\ell_p$ spaces, 
\cite{lindenstrauss1963}]\label{thm:lindenstrauss}
We have
\[
\rho_{\ell_p}(\tau) = 
\left\{
\begin{array}{ll}
	(((1+\tau)^p + (1-\tau)^p)/2)^{1/p}-1 \leq 2^p\tau^2, \text{ 
if } 2 \leq p < \infty\\
	(1+\tau^p)^{1/p}-1 \leq \tau^p/p, \text{ if } 1 	\leq p \leq 2
\end{array}
\right.
\]
\end{theorem}

\begin{proof}
By \cite[end of Section~2]{lindenstrauss1963}, we only need to show 
$(((1+\tau)^p + |1-\tau)^p)/2)^{1/p}-1 \leq 2^p\tau^2$ for $\tau\in(0,1)$ and $2\leq p<\infty$.

By computing $\frac{\di}{\di p}\left[ (1+\tau)^p + (1-\tau)^p \right]$, and then 
$\frac{\di}{\di \tau}\frac{\di}{\di p}\left[ (1+\tau)^p + (1-\tau)^p \right]$, one obtains that
$(1+\tau)^p + (1-\tau)^p\leq (1+\tau)^{\ceil{p}} + (1-\tau)^{\ceil{p}}$. Next, 
by taking the binomial expansion, one checks that 
$\left[(1+\tau)^{\ceil{p}} + (1-\tau)^{\ceil{p}}\right]\leq \left(1+2^p\tau^2\right)^{\floor{p}}$,
completing the proof.
\end{proof}

Theorems~\ref{cov:smooth} and \ref{thm:lindenstrauss}
imply the following.
\begin{corollary}[\twocovs for $\ell_p$ balls]\label{cor:twocovlp}
For small enough $\epsilon$, there exists a \twocov for $\ell_p$ 
balls using $2^{O(n)}\log(1+1/\epsilon)(\frac{1}{\epsilon})^{(n/2)}$ convex 
bodies for $2 \leq p < \infty$ and\\ 
$2^{O(n)}\log(1+1/\epsilon)(\frac{1}{\epsilon})^{(n/p)}$ convex bodies for $1 
\leq p \leq 2$.
\end{corollary}

\section{Using \texorpdfstring{$(2,\epsilon)$}{(2,epsilon)}-coverings for the 
Closest Vector Problem}\label{sec:cvpalg}

We first recall the goal and some important notions of this section: We are 
given a rational lattice $\Lambda(A) = \{Ax \st x \in 
\Z^n\}$, with $ A \in \Q^{n \times n}$ and a target vector $t \in \Q^n$, and we 
would like to solve $(1+\epsilon)$-approximate \cvp[K], i.e. find a lattice 
vector $v \in \Lambda(A)$ such that $\|v-t\|_K \leq (1+\epsilon)\min_{w\in 
\Lambda(A)}\|w-t\|_K$. $\|\cdot\|_K$ is defined by $\|x\|_K = \inf\{s \st x 
\in sK\}$, if $K$ is origin symmetric and convex, this defines a norm. If $0$ 
is not the center of symmetry but in the interior of $K$ then we lose the 
symmetry, i.e. $\|x\|_K \neq \|-x\|_K$. We denote by $b$ the encoding length 
of the relevant input: $A$, $t$, $\epsilon$, encoding length of $K$, etc.

In this section, we will first describe how a \twocov for 
$K$ using $N$ convex bodies boosts any \acvp[2]{} solver for general 
norms to a \acvp{K} solver at the expense of a factor $N 2^{O(n)}\poly(b, 
\frac{1}{\epsilon})$ in the running time. This algorithm, together 
with the construction of Propositions~\ref{covering:zonotope} and 
\ref{covering:polytope} directly implies a \acvp{} solver 
for polytopes and zonotopes with running time of
$2^{O(n+m)}(\log(1+1/\epsilon))^{m}$ times some polynomial in $b$ and $n$ and 
with space requirement that of the 
\acvp[2]{} solver used.

Next, we are going to adapt the construction of Theorem~\ref{cov:smooth} to 
yield a randomized algorithm, that for some fixed point $p \in K$, generates a 
local \twocov for $K$ containing $p$. This yields a randomized 
\acvp{} solver 
with the improved running time for $\ell_p$ norms and with space requirement 
only depending on that of the $2$-approximate \cvp{} solver used. This 
construction can also be derandomized.

The boosting procedure we are going to describe assumes that we are able to 
sample uniformly within $K$ and that we can calculate a separating hyperplane 
at any point on the boundary of $K$. However, if only a weak membership and a 
weak separation oracle is provided, the procedure can be adapted such that it 
suffices to sample almost uniformly, see the algorithm of Dyer, Frieze and 
Kannan \cite{DBLP:journals/jacm/DyerFK91}, and to only calculate a weakly 
separating hyperplane. We neglect this implementation detail.

As for the convex body $K$, we assume that
$n^{-3/2}B_2^n \subseteq K \subseteq B_2^n$, and thus,
\begin{equation}\label{eq:Kboundedabovebelow}
\|x\|_2 \leq \|x\|_K \leq n^{3/2}\|x\|_2.
\end{equation}
This can be ensured by applying an affine transformation, which is 
polynomial in the input size of $K$, to both $K$ and the lattice $\Lambda(A)$, 
see \cite{GroetschelLovaszSchrivjer88}.

For concreteness, we choose to use the elegant and currently fastest 
algorithm for general norms by Dadush and Kun as our 
\acvp[2]{} solver.
\begin{theorem}[Approximate CVP in any norm \cite{DBLP:journals/toc/DadushK16}]
\label{solver:dadushkun}
There exists a deterministic algorithm that for any 
norm $\norm{\cdot}_K$, $n$-dimensional lattice $\Lambda(A)$ and for any target 
$t\in \R^n$, computes $y \in \Lambda(A)$, a $(1+\epsilon)$-approximate 
minimizer to $\norm{t-x}_K, x \in \Lambda(A)$, in time 
$O(\poly(n,b)2^{O(n)}(1+\frac{1}{\epsilon})^n)$ and $O(\poly(n,b)2^n)$ 
space.
\end{theorem}

\begin{theorem}[Boosting 2-CVP using a 
$(2,\epsilon)$-covering]\label{2-epsilon:det-solver}
Assume we are given an origin symmetric convex body $K$ in $\Ren$ and a
\twocov for $K$ consisting of $N$ convex bodies. Then we can 
solve the 
\acvp[(1+7\epsilon)]{K} for $\Lambda(A)$ and target $t\in 
\Q^n$ with $O\left(N\log(1+\frac{1}{\epsilon})(\log(n) + \log(b))\right)$ calls 
to a 2-approximate \cvp{} solver for general norms.


\end{theorem}
\begin{proof}
We may multiply $\Lambda(A)$ and $t$ by the least 
common multiple of the denominators of the $n^2$ entries of $A$ and the $n$ entries of $t$. The 
resulting lattice and target are integral, $\Lambda(\tilde{A}) \in \Z^{n \times 
n}$ and $ \tilde{t} \in \Z^n$. Since 
the lowest common multiple is bounded by $2^{(n^2 + n)b}$, the resulting basis 
of $\tilde{A}$ has Euclidean length at most $2^{(n^2 + n)b}$. Assuming $t 
\notin \Lambda(A)$, we see that
\[1 \leq \min_{x \in \Lambda(\tilde{A})}\|x-\tilde{t}\|_2 \leq n2^{(n^2 + 
n)b}.\]
By our assumption \eqref{eq:Kboundedabovebelow}, we have
\[1 \leq \min_{x \in \Lambda(A)}\|x-t\|_K \leq n^{5/2}2^{(n^2 + n)b}.\]

Let $\{Q_i+c_i\}_{i=1}^N$ be the given \twocov for $K$, where 
the origin is the centroid of each of the $Q_i$.

For our algorithm, for any norm $\|\cdot\|_Q$, we assume that the 
2-approximate \cvp[Q] algorithm that we use with target $t$ only returns a 
lattice 
vector $v$ if $\|t-v\|_Q \leq 2$.

We want to find $f$ such that $c_i+(1+\epsilon)^f Q_i$ contains a lattice 
vector for some $i\in[N]$, but $c_i+(1+\epsilon)^{f-1} Q_i$ contains no lattice 
vector for any $i\in[N]$. As in \cite{DBLP:conf/compgeom/EisenbrandHN11}, we 
apply a binary search for $f$.
\begin{enumerate}
\item Initialize $L \leftarrow 0$, $U \leftarrow \left\lceil 
\log_{1+\epsilon}n^{5/2} 2^{(n^2 + n)b}\right\rceil$ and $x = 0$
\item While $U-L \geq 4$,  do a binary search step:
\smallskip
\begin{enumerate}
\item For all $i \in [N]$, solve a 2-approximate \cvp[(1+\epsilon)^{L + 
\lceil(U-L)/2\rceil}Q_i] problem with target $(1+\epsilon)^{L + 
\lceil(U-L)/2\rceil}c_i + t$
\item If some lattice vector $v$ is returned, update $U \leftarrow \lceil 
\log_{1+\epsilon}\|v-t\|_K \rceil$ and $x \leftarrow v$.
\item Otherwise, update $L \leftarrow L+\lceil(U-L)/2\rceil$
\end{enumerate}
\smallskip
\item Return $x$.
\end{enumerate}
\smallskip
It is immediate that for any $\lambda 
>0$, $\{\lambda Q_i+\lambda c_i\}_{i=1}^N$ is a 
\twocov for $\lambda K$. Thus if, for some $L$ and $U$ at step 
$2(b)$, no lattice vector $v$ is returned, then 
\begin{alignat*}{1}
t + (1+\epsilon)^{L + \lceil(U-L)/2\rceil}K \subseteq t + \bigcup_{i=1}^N 
(1+\epsilon)^{L + \lceil(U-L)/2\rceil}(c_i + Q_i)
\end{alignat*}
contains no lattice vector, and so 
$\min_{v \in \Lambda(A)}\|v-t\|_K \geq (1+\epsilon)^{L + \lceil(U-L)/2\rceil}$. 

In the case a lattice vector is returned, then
$$\min_{x \in \Lambda(A)}\|t-x\|_K \leq \|v-t\|_K \leq (1+\epsilon)^{L + 
\lceil(U-L)/2\rceil + 1}$$ since the $Q_i$ are a \twocov of $K$. 
Since $U$ and $L$ are valid upper and lower bounds for $f$ at the beginning of 
the algorithm, we see 
that throughout the algorithm, the following invariant is maintained:
$$(1+\epsilon)^L \leq \min_{v \in \Lambda(A)}\|v-t\|_K \leq (1+\epsilon)^U.$$
If the algorithm terminates, then $U-L \leq 3$ since $U$ and $L$ are both 
integers. Thus, because of the above invariant, the lattice vector $x \in 
\Lambda(A)$ returned satisfies
\[\|x-t\|_K \leq (1+\epsilon)^U \leq (1+\epsilon)^{L+3} \leq (1+\epsilon)^3 
\min_{v \in \Lambda(A)}\|v-t\|_K \leq (1+7\epsilon)\min_{v \in 
\Lambda(A)}\|v-t\|_K.\]

It remains to be shown that the binary search terminates in 
$O(\frac{1}{\epsilon}(\log(n) + \log(b))$ steps. Indeed, for some $U$ and $L$, 
let 
$U_{new}$, $L_{new}$ be the $U$ and $L$ after having executed step $2$ once. If 
$U-L \geq 6$, it is straightforward to check that $U_{new}-L_{new} \leq 
\frac{3}{4}(U-L)$. If $4 \leq U-L \leq 5$, $U_{new}-L_{new} \leq (U-L)-1$. 
Since $U-L \leq \log_{1+\epsilon}(n^{5/2}2^{(n^2 + n)b})$ at the beginning of 
the algorithm, we are done after $\log_{5/4}(\log_{1+\epsilon}(n^{5/2}2^{(n^2 + 
n)b})) = O(\log(1+\frac{1}{\epsilon})(\log(n) + \log(b)))$ iterations.
\end{proof}

\begin{corollary}[$(1+\epsilon)$-approximate $\text{CVP}$ for polytopes and 
zonotopes]\label{cor:algforpolytopeszonotopes}
Let $K$ be a full-dimensional origin symmetric polytope with $m$ facets or a 
full-dimensional zonotope with $m$ 
generators (in particular, $m\geq n$). Then for any $\epsilon \in (0,1)$, the $(1+\epsilon)$-approximate 
\cvp[K] problem can be solved deterministically in time 
$O(\poly(n,b, \frac{1}{\epsilon})2^{O(n+m)}\log(1+1/\epsilon)^{m})$ and space 
$O(\poly(n)2^n)$.
\end{corollary}
\begin{proof}
Replace $\epsilon$ by $\epsilon/7$ and run the algorithm in 
Theorem~\ref{2-epsilon:det-solver} on a \twocov of $K$
constructed in the proof of Proposition~\ref{covering:zonotope} or 
\ref{covering:polytope}.
To avoid a space requirement depending on the number of convex bodies $N$ 
required in the \twocov for $K$, every time we call step 
$2(a)$ of the algorithm, for each $i \in [N]$, we first calculate $Q_i$ and 
then run the appropriately scaled 2-approximate \cvp{} instance. 

\end{proof}

\begin{remark}
	The preceding corollary is the reason why we 
opted to describe a \twocov with symmetric convex bodies for 
symmetric polytopes in Proposition~\ref{covering:polytope}: The algorithm of 
Dadush and Kun can handle non-symmetric norms $\|\cdot\|_K$, provided $0$ is in 
some sense "close" to the centroid of $K$, for more details see 
\cite{DBLP:journals/toc/DadushK16}. Since calculating deterministically the 
centroid is a hard problem and no efficient algorithms are known, see 
\cite{DBLP:conf/compgeom/Rademacher07}, we would most likely have to resort to 
a randomized algorithm to approximate the centroid which in turn randomizes our 
boosting procedure. 
\end{remark}

\begin{theorem}[Local $(2,\epsilon)$-covering]
\label{theorem:local_covering}
Let $K$ be an origin symmetric convex body such that $\|\cdot\|_K$ has modulus 
of smoothness $C\tau ^q$ for $C, q > 1$ and $\epsilon \in (0,1)$. Then, in polynomial 
time, we can 
find at most $O(\log(1+1/\epsilon))$ origin symmetric convex bodies $\{Q_i\}$ and 
translations $\{c_i\}$ such that for some constant $c > 0$:
\begin{enumerate}
\item For all $i$, $c_i + 2Q_i \subseteq (1+\epsilon) K$.
\item For $q \in K$, the probability that $q$ is contained in $c_i + Q_i$ for 
some $i$ is greater than $\min(2^{-cn}C^{-n/q}(1/\epsilon)^{n/q}, (\frac{1}{8^q 
C})^{n/(q-1)})$
\end{enumerate}
\end{theorem}

\begin{proof}
Set $\epsilon \leftarrow \epsilon/3$.
If $\epsilon > \left(\frac{1}{8C^{1/q}}\right)^{q/(q-1)}$, we uniformly sample 
a 
point $x$ from $(1+\epsilon)K$ and return $\epsilon K$ and $x$. Any point in 
$K$ has probability greater or equal than 
$$\left(\frac{\epsilon}{1+\epsilon}\right)^n$$ 
of being covered by $x + \epsilon K$.

If $\epsilon \leq \left(\frac{1}{8C^{1/q}}\right)^{q/(q-1)}$, similar as in 
Theorem~\ref{cov:smooth}, 
we set $\delta = \frac{1}{4} \left(\frac{\epsilon}{C}\right)^{1/q}$. We 
uniformly sample a 
point $x$ from $(1+\delta/4)K$. Let $p = \frac{x}{\|x\|}$ and for $i \in 
[\log(1/\epsilon)]$, consider the slices $C_p(i)$ of $C_p$ as in 
\eqref{mod.cylinders} in the proof of Theorem~\ref{cov:smooth}.

For all such $C_p(i)$, denoting by $c(C_p(i))$ its centroid, we return the 
origin 
symmetric convex bodies $\{C_p(i)-c(C_p(i))\}$ and the translations 
$\{c(C_p(i))\}$.

Next, fix a point $q \in K$. With probability greater 
or equal to 
\[\frac{1}{2}\frac{(\delta/4)^n}{(1+\delta/4)^n} \text{ we have that } 
\norm{\frac{q}{\norm{q}} - x} \leq \delta/4.\]
In that case, $\norm{\frac{q}{\norm{|q}}- p} \leq \delta/2 \leq \delta - 
\epsilon$ and 
thus, $C_p$ as in \eqref{mod.cylinders} of Theorem~\ref{cov:smooth} contains 
$q$.
It follows that for some $c > 0$ independent of $n, C$ and $q$, with 
probability greater or equal to
\[ 2^{-cn}C^{-n/q}\epsilon^{n/q} \]
one of the cylinders $C_p(i)$ contain $q$.
\end{proof}

The next theorem combines the algorithms of 
Theorems~\ref{theorem:local_covering} and 
\ref{2-epsilon:det-solver} to yield an efficient $(1+\epsilon)$-approximate 
CVP solver for norms with a well bounded modulus of smoothness. 

\begin{theorem}[Boosting 2-CVP for a body with small modulus of 
smoothness]\label{thm:algmodsmoothness}
Let $K$ be a origin symmetric convex body with modulus of smoothness 
\[\rho_K(\tau) 
\leq C\tau^q, \text{ with } C,q > 1\] Then the algorithm presented in the proof 
solves \acvp{K} with 
probability at least $1-2^{-n}$. Its running time is 
$O(\poly(n,b,\log(1/\epsilon))(2^{O(n)} C^{n/q}\left(1/\epsilon\right)^{n/q}+ 
O(C)^{n/{(q-1)}}))$, 
and the space requirement is equal to that of a \acvp[2]{} solver that handles 
any norm.
\end{theorem}
\begin{proof}
We set $\epsilon \leftarrow \epsilon/7$ and without loss of generality, we may 
assume
\[1 \leq \min_{x \in \Lambda(A)}\|x-t\|_K \leq n^{5/2}2^{(n^2 + n)b}.\]
We again assume that, for any norm $\norm{\cdot}_Q$, the \acvp[2]{Q} with 
target $t$ only returns a lattice vector $v$ if $\|t-v\|_Q \leq 2$, if there is 
no such $v$, it returns nothing.

We adapt the algorithm of Theorem~\ref{2-epsilon:det-solver}:
\begin{enumerate}
\item Initialize $L \leftarrow 0$, $U \leftarrow \left\lceil 
\log_{1+\epsilon}n^{5/2} 2^{(n^2 + n)b}\right\rceil$ and $x = 0$
\item While $U-L \geq 4$,  do a binary search step:
\smallskip
\begin{enumerate}
\item Run the algorithm from Theorem \ref{theorem:local_covering} and denote 
the returned convex bodies and translations by $Q_i$ and $c_i$ respectively. 
For all $i$, solve a 2-approximate \cvp[(1+\epsilon)^{L + 
\lceil(U-L)/2\rceil}Q_i] 
problem with target $(1+\epsilon)^{L + \lceil(U-L)/2\rceil}c_i + t$. Repeat $N$ 
times.
\item If some lattice vector $v$ is returned, update $U \leftarrow \lceil 
\log_{1+\epsilon}\|v-t\|_K \rceil$ and $x \leftarrow v$.
\item Otherwise, update $L \leftarrow L+\lceil(U-L)/2\rceil$
\end{enumerate}
\smallskip
\item Return $x$.
\end{enumerate}
Correctness of the algorithm follows from Theorem~\ref{2-epsilon:det-solver}, 
provided step 2 runs correctly (i.e. correctly detects whether there is a 
lattice point or not with high probability) for all 
$O(\log(\frac{1}{\epsilon})(\log(n) + \log(b)))$ iterations. To verify this, 
let $v \in \mathcal{L}$ be some lattice vector contained in a homothet of $K$ 
at 
some fixed iteration of the algorithm. With probability $p = 
2^{-cn}C^{-n/q}(1/\epsilon)^{n/q}$ or $(\frac{1}{8^q C})^{1/(q-1)}$ 
respectively, one of the convex bodies returned by one 
run of Theorem~\ref{theorem:local_covering} contains $v$. Thus, repeating step $2(a)$ $n(2^{cn}C^{n/q}(1/\epsilon)^{n/q} + (8^q 
C)^{1/(q-1)})$ 
times, with probability 
greater 
than $1-2^{-n}$, $v$ is contained in one of the convex bodies returned and step 
$2$ runs correctly. Since step $2$ needs to run correctly each of the $ 
O(\log(\frac{1}{\epsilon})(\log(n) + \log(b)))$ iterations necessary to find 
the correct $U$ and $L$, by the union bound, it is sufficient to set $N = 
O(n\log(\log(\frac{1}{\epsilon})(\log(n) + 
\log(b)))2^{cn}C^{n/q}(1/\epsilon)^{n/q} + (8^q C)^{1/(q-1)})$ to guarantee a 
success probability 
of $1-2^{-n}$.
This implies the bound on the running time.
\end{proof}

In our proof of Theorem~\ref{thm:algmodsmoothness}, instead of applying our 
local covering algorithm, Theorem~\ref{theorem:local_covering}, we could use a 
recent result of Dadush \cite[Theorem~4.1]{Da2013}. There, a deterministic 
algorithm is presented to build and iterate over an epsilon net in 
$2^{O(n)}(1+1/\epsilon)^n$ time and $\poly(n)$ space. For symmetric convex bodies 
with modulus of smoothness bounded by $C\tau^q$, we may apply this result with 
$O\left(\epsilon^{1/q}\right)$, as in Theorem ~\ref{theorem:local_covering}, in 
place 
of $\epsilon$ to build a covering of size 
$O(\frac{1}{\epsilon})^{n/q}$. This would replace the sampling part 
in Theorem~\ref{theorem:local_covering} and thus derandomizes our boosting 
procedure.

\begin{remark}\label{rem:nonsymm}
One may consider convex bodies that are not necessarily origin symmetric. 
Assume 
that a 
convex body $K$ is \emph{$\gamma$-symmetric}, that is, $\vol(K \cap 
-K) \geq \gamma^n\vol(K)$. Then the result of Dadush and Kun 
(Theorem~\ref{solver:dadushkun}) still applies (see 
\cite{DBLP:journals/toc/DadushK16}), and it is straightforward to modify the 
above algorithm to obtain a $(1+\epsilon)$-approximate CVP algorithm for 
$\norm{\cdot}_K$ using $2^{O(n)}(\frac{1}{\gamma \epsilon})^n$ calls 
to a $2$-approximate CVP algorithm handling any symmetric norm, for instance 
the AKS based algorithm of 
Dadush \cite{DBLP:conf/latin/Dadush12}, resulting in an algorithm with time 
$O(\frac{1}{\gamma \epsilon})^n$ and space $2^{O(n)}$. We essentially 
use Theorem~\ref{theorem:local_covering} with $q = 1$: we sample a point $p$ 
in $(1+\epsilon/3)K$ and output $\frac{\epsilon}{3}(K \cap -K)$ and $p$. 
Thus, each point in $K$ has probability greater or equal to 
$2^{-O(n)}(\frac{1}{\gamma \epsilon})^n$ of being covered.
\end{remark}

\section{Sparsifiers and the modulus of smoothness}\label{sec:sparsifier}
In this section we describe a surprising connection between lattice 
sparsifiers as used by Dadush and Kun and the modulus of smoothness. 
Informally, our main technical contribution is the observation that for a  
lattice-point-free convex body $K$ with modulus of smoothness bounded by 
$C\tau^q$, a $O(\epsilon^{1/q})$-sparsifier for $K$ preserves the metric 
information up to an additive error of $O(\epsilon)$. 
We will show that we can tweak the algorithm of Dadush and Kun using this 
simple observation in order to match the running time of the preceding 
boosting procedure. 

We will only consider origin symmetric-convex bodies $K\subseteq \Ren$.
\begin{definition}[Lattice sparsifier for origin symmetric $K$, 
\cite{DBLP:journals/toc/DadushK16}]
	Let $K \subseteq \Ren$ be an origin-symmetric convex body, 
$\mathcal{L}$ 
be a $n$-dimensional lattice and $\delta > 0$. A $(K,\delta)$ sparsifier for 
$\mathcal{L}$ is a sublattice $\mathcal{L'}\subseteq \mathcal{L}$ satisfying
	\begin{enumerate}
		\item $G(K, \mathcal{L'}) \leq O(\frac{1}{\delta})^n$
		\item $\forall x \in \Ren, \, d_K(\mathcal{L'},x) \leq 
d_K(\mathcal{L},x) + \delta$,
	\end{enumerate}
\end{definition}
where $G(K,\mathcal{L})$ denotes the maximal number of lattice vector any 
translate 
of $K$ can contain, formally:
$$G(K, \mathcal{L}) = \max_{x \in \Ren}|(K+x)\cap \mathcal{L}|.$$
By a covering argument (see Lemma 2.3 \cite{DBLP:journals/toc/DadushK16}), 
$G(dK,\mathcal{L})\leq (2d+1)^n G(K,\mathcal{L})$. 
By the second condition, if $\mathcal{L'}$ is a $(K,\delta)$-sparsifier for 
$\mathcal{L}$, for every lattice point $v \in \mathcal{L}$, there is $v' \in 
\mathcal{L}'$ such that $\|v-v'\|_K \leq \delta$. These two conditions ensure 
that the resulting lattice $\mathcal{L'}$ is thinned out according to the 
geometry of $K$: the first condition guarantees that $K$ (or a dilate of $K$) 
cannot contain too many lattice vectors of $\mathcal{L'}$ (hence enumeration is 
not too costly), but, by the second condition, $\mathcal{L'}$ is rather close 
to $\mathcal{L}$ and thus serves as a good approximation.

We now come to the main observation:
\begin{lemma}\label{lemma:observation}
Let $K$ be an origin symmetric convex body with modulus of smoothness 
bounded by $\rho_K \leq C\tau^q$, $q \geq 1$, $\mathcal{L}$ a lattice and 
$t\in\Ren$ 
a target vector. Assume that $t + K$ does not contain any lattice vector $v \in 
\mathcal{L}$ in its interior. Let $\mathcal{L'}$ be a $(K, \epsilon^{1/q})$ 
sparsifier for $\mathcal{L}$. Then
\[d_K(\mathcal{L'}, t) \leq d_K(\mathcal{L},t) + 2C\epsilon.\]
\end{lemma}
\begin{proof}
Denote by $v \in \mathcal{L}$ a closest lattice vector to $t$, and set 
$R:=d_K(\mathcal{L},t)$. Clearly, $R=\norm{v-t}_K\geq1$. By the second 
condition of the sparsifier, there is a lattice vector $w \in \mathcal{L'}$ with
$\norm{w-v}_K\leq \epsilon^{1/q}$. Denoting by $y:=w-v \in \mathcal{L}$, 
the definition of the modulus of smoothness yields
\begin{alignat*}{1}
\norm{\frac{w-t}{R}}_K=
\norm{\frac{v-t}{R}+\frac{y}{R}}_K
\leq
2 + 2C\epsilon/R^q-
\norm{\frac{v-t}{R}-\frac{y}{R}}_K
\leq
1 + 2C\epsilon/R^q,
\end{alignat*}
where we used the fact that $v-y\in\mathcal{L}$, and hence, 
$\norm{(v-y)-t}_K\geq R$. Multiplying the inequality by 
$R$ and observing that $R,q\geq 1$ completes the proof of 
Lemma~\ref{lemma:observation}.
\end{proof}

Next, we present the algorithmic application of the previous lemma to the 
$(1+\epsilon)$-approximate Closest Vector Problem under a symmetric norm. We 
adopt the same notation as in Section~\ref{sec:cvpalg}. We may assume that $t 
\in 
\mathbb{Z}^n$, $\mathcal{L}(A) \subseteq \mathbb{Z}^{n}$ and 
$\|t\|_{\infty}, 
\|A\|_{\infty} \leq 2^{(n^2 + n)b}$. We assume $n^{-3/2}B_2^n \subseteq K 
\subseteq \frac{1}{2}B_2^n$. Thus, $d_K(\mathcal{L},t)\leq 2n^{5/2}2^{(n^2 + 
n)b}$, and, if $t \notin \mathcal{L}(A)$, $t + K$ does not contain a lattice 
vector. We will need the following two algorithms.

\begin{theorem}[\lattenum($K, t, \mathcal{L},\epsilon)$, 
\cite{Dadush:2011:ELA:2082752.2082900}]
	Let $\mathcal{L}(A)$ be a lattice, $K$ a convex body in $\Ren$ and 
$\epsilon 
>0$. 
There is a deterministic algorithm that outputs all $S$ such that
	$$(t + K) \cap \mathcal{L} \subseteq S \subseteq (t+K + \epsilon 
B_2^n)\cap \mathcal{L}$$
	in time $G(K,\mathcal{L})2^{O(n)}\poly(n,b)$ and $2^n\poly(n,b)$ space.
\end{theorem}

\begin{theorem}[\lattspars($\mathcal{L}(A), K, \delta$), 
\cite{DBLP:journals/toc/DadushK16}]
	For $\delta > 0$, a basis $A'$ for a $(K,\delta)$-sparsifier for 
$\mathcal{L}(A)$ can be computed deterministically in $2^{O(n)}\poly(n,b)$ time 
and 
$2^n\poly(n,b)$ space.
\end{theorem}

We now combine these two theorems with Lemma~\ref{lemma:observation}.
\begin{theorem}\label{thm:sparsificationmodulus}
	There is an algorithm (described in the proof) that for an origin 
symmetric convex body $K$ in $\Ren$, with modulus of smoothness bounded by 
$\rho_K \leq 
C\tau^q$ with some $C,q \geq 1$, solves $(1+\epsilon)$-\cvp[K] for any lattice 
$\mathcal{L}$ and target vector $t\in\Ren$ in time 
$O(\frac{C}{\epsilon})^{n/q}\poly(n,b)$ and 
space $2^n \poly(n,b)$.
\end{theorem}

\begin{proof}[Proof of Theorem~\ref{thm:sparsificationmodulus}]
We may assume $\epsilon \leq 1$. If $t \in \mathcal{L}(A)$ (this can be 
checked in $\poly(nb)$ time), return $t$. Else, set $\bar{\epsilon} = 
\frac{\epsilon}{4C}$ and $d = 0$ and apply the following algorithm.

\smallskip
\begin{minipage}{0.96\textwidth}
\begin{enumerate}
 \item Set $K_d = 2^d K$. 
 \item Apply \lattspars$(K_d,\mathcal{L}, 
   \bar{\epsilon}^{1/q}$). Denote the sparsified lattice by $\mathcal{L'}$.
 \item Apply \lattenum$((2+\epsilon)K_d, t,\mathcal{L'},\epsilon)$. 
If there is a lattice vector in 
$t + (2+\epsilon)K$, return the closest one to 
$t$, and stop. 
Else, set $d \leftarrow d+1$ and go to $(1)$.
\end{enumerate}
\end{minipage}
\smallskip

Let $k$ be the largest positive integer such that $t + K_{k}$ does not 
contain a lattice vector. First, we claim that the 
algorithm will terminate at iteration $d\leq k$.
Indeed, since $t+2K_k=t+K_{k+1}$ contains a lattice vector of $\mathcal{L}$, by 
Lemma~\ref{lemma:observation}, $(2+\epsilon)K_k$ contains a lattice vector of 
$\mathcal{L'}$, and hence, the algorithm will terminate at $d=k$, or before.

To bound the error, we assume that the algorithm terminated at iteration $d$.
By the previous paragraph, $t + K_{d}$ does not contain a lattice vector, and 
thus,
\begin{equation}\label{eq:kdempty}
d_K(\mathcal{L}, t) \geq 2^{d}.
\end{equation}

Let $v$ denote the lattice vector returned by \lattenum$((2+\epsilon)K_d, 
t,\mathcal{L'},\epsilon)$. By Lemma~\ref{lemma:observation}, we only have an 
additive error of $2C\bar{\epsilon}=\frac{\epsilon}{2}$ with respect to 
$\norm{\cdot}_{K_d}$, that is,
\[
d_{K_d}(\mathcal{L}, t)\leq \norm{t-v}_{K_d}+\frac{\epsilon}{2},
\]
which, by \eqref{eq:kdempty} yields
\[
d_{K}(\mathcal{L}, t)
\leq
\norm{t-v}_{K}+2^d\epsilon
\leq
\norm{t-v}_{K}+\epsilon d_{K}(\mathcal{L}, t),
\]
and hence, 
$d_{K}(\mathcal{L}, t)\leq \frac{1}{1-\epsilon/2}\norm{t-v}_{K} 
\leq(1+\epsilon)\norm{t-v}_{K}$. Thus, we found a 
$(1+\epsilon)$-approximate solution.

Next, we consider the time and space requirements. It is clear that step $(2)$
always takes time $2^{O(n)}\poly(n,b)$ and space $2^n\poly(n,b)$, independently 
of $d$. 
Note that $G((2+\epsilon)K, \mathcal{L'}) \leq G(3K,\mathcal{L'})\leq 
O(\frac{1}{\bar{\epsilon}})^{n/q}$, and thus, step $(3)$ takes
$O(\frac{C}{\epsilon})^{n/q}\poly(n,b)$ time and $2^n \poly(n,b)$ space. 
Since $d_K(t,\mathcal{L})\leq 2n^{5/2}2^{(n^2 
+ n)b}$, we need at most $\log_{2}(2n^{5/2}2^{(n^2 + n)b}) = \poly(n,b)$ 
iterations, resulting in time $O(\frac{C}{\epsilon})^{n/q}\poly(n,b)$. 
This completes the proof of Theorem~\ref{thm:sparsificationmodulus}.
\end{proof}

\section{Appendix: Proof of some lemmas}\label{sec:appendix}

\begin{proof}[Proof of Lemma~\ref{lemma:symmetric}]
	We cover $K$ greedily by copies of $\frac{\epsilon}{2}K$ as follows. If 
after 
	selecting $i-1$ homothetic copies of $K$ there is a point $p_i \in K$ 
not 
	yet covered, we take $Q_i=p_i+\frac{\epsilon}{2}K$. To see that after 
$N \leq 
	(\frac{5}{\epsilon})^n$ steps, all 
	points of $K$ are covered, we notice that the sets $\frac{1}{2}\odot 
Q_i$ are 
	non-overlapping, and are contained in $(1+\epsilon/4)K \subseteq 
\frac{5}{4}K$. Taking the volume of 
	these 
	sets, we obtain the desired bound.
\end{proof}

\begin{proof}[Proof of Lemma~\ref{lem:symmetriccovering}]

Let $\{Q_i\}_{i=1}^N$ be a \twocov of $K$. For each $i\in[N]$, 
we will find a \twocov[(2,1)] for $Q_i$ using at most $10^n$ centrally 
symmetric convex bodies. Thus, the union of these at most $10^nN$ symmetric 
sets will yield a \twocov of $K$.
Fix $i\in[N]$ and set $\tilde Q_i=\frac{1}{2}\left((Q_i-c(Q_i))\cap 
(c(Q_i)-Q_i) \right)$. 
In the same fashion as in the proof of Lemma~\ref{lemma:symmetric}, let 
$\{b_1,\ldots,b_m\}$ be a maximal subset of $Q_i$ such that the interiors of 
the 
sets $b_1+\frac{1}{2}\tilde Q_i, \ldots, b_m+\frac{1}{2}\tilde Q_i$ are 
pairwise 
disjoint. Clearly, $\tilde Q_i+\{b_1,\ldots,b_m\}$ is a covering of $Q_i$.

By a result of Milman and Pajor \cite{journalMilmanP}, if the centroid of a 
convex body $Q$ 
in $\Ren$ is the origin, then 
\begin{equation}\label{eq:MiPa}
 \vol(Q\cap -Q) \geq 2^{-n}\vol(Q).
\end{equation}
Thus, $\vol(b_k + \frac{1}{2} \tilde Q_i)\geq 8^{-n}\vol(Q)$ ($k=1,\ldots,m$). 
Since $b_k\in Q_i$ and $\frac{1}{2}\tilde Q_i\subseteq \frac{1}{4}(Q_i-c(Q_i))$, 
we have that $b_k + \frac{1}{2} \tilde Q_i\subseteq\frac{5}{4}\odot Q_i$. Thus, 
$m\leq 10^n$.

To see that $\tilde Q_i+\{b_1,\ldots,b_m\}$ is \twocov[(2,1)] of $Q_i$, note 
that $2\tilde Q_i\subseteq (Q_i-c(Q_i))$, and hence $b_k+2\tilde Q_i\subseteq 
Q_i+ (Q_i-c(Q_i)) = 2\odot Q_i$, as required.
\end{proof}

\begin{proof}[Proof of Lemma~\ref{lemma:symmetrizer}]
The same argument as that used in the proof of Lemma~\ref{lemma:symmetric} 
combined with \eqref{eq:MiPa} yields it.
\end{proof}

\begin{proof}[Proof of Proposition~\ref{covering:zonotope}]
We may assume that $\epsilon = (2^k -1)^{-1}$ for some positive integer $k$.

For $i \in [k]$, the following union of translated intervals is a \twocov of 
$[-b,b]$:
\begin{alignat*}{1}
[-b,b] \subseteq \bigcup_{\delta \in \{\pm 1\}, \, j \in 
[k]}\left(\delta(1- (2^j-1)\epsilon)b + 
[-2^{j-1}\epsilon b, 2^{j-1}\epsilon b]\right)
\end{alignat*}
We may decompose analogously every line segment generating $\mathcal{Z}$ and 
combine them to give a \twocov for $\mathcal{Z}$:
\begin{alignat*}{1}
\mathcal{Z} \subseteq \bigcup_{\delta \in \{\pm 1\}^m, \, \alpha \in 
[k]^m} 
\sum_{i=0}^{k}\left(\delta_i(1-(2^{\alpha_i}-1)\epsilon)b_i + 
[-2^{\alpha_i-1}\epsilon b_i, 2^{\alpha_i-1}\epsilon b_i]\right)
\end{alignat*}
This is a \twocov for $\mathcal{Z}$ using 
$(2\log_2(1+1/\epsilon)+1)^m$ (translated) zonotopes.
\end{proof}

\begin{proof}[Proof of Proposition~\ref{covering:polytope}]
We may assume that $\epsilon = \left((4/3)^k-1\right)^{-1}$ for some positive 
integer $k$.

For $\alpha \in [k]^m$ and $\delta \in \{\pm 1\}^m$, 
consider the following polytopes:
\[
 \bar{Q}(\alpha, \delta) =
\]
\begin{alignat*}{1}
 \left\{x \st 
\left(1-\left(\left(\frac{4}{3}\right)^{\alpha_i}-1\right)\epsilon\right)b_i 
\leq \delta a_i^T x 
\leq 
\left(1-\left(\left(\frac{4}{3}\right)^{\alpha_i-1}-1\right)\epsilon\right)b_i\,
 , i \in 
[m]\right\}
\end{alignat*}
For each facet direction $|a_i^T x|\leq b_i$, scaling each of the resulting 
(non-empty) $\bar{Q}$ around any point in its interior by a factor $4$, it is 
straightforward to check that the resulting convex body is contained inside 
$\{x \in \Ren \st |a_i^T x| \leq (1+\epsilon)b_i\}$. It follows that each such 
non-empty polyhedron $\bar{Q}$ can be scaled by a factor $4$ 
around any point in it and the resulting polytope is still contained inside 
$(1+\epsilon)P$ and it is clear that $P$ is contained in the union of the 
$\bar{Q}(\alpha, \delta)$. 

We could stop here and have a \twocov for $P$, but we are not guaranteed that 
the resulting cells are centrally symmetric. In order to ensure this, we will 
symmetrize the resulting $\bar{Q}(\alpha, \delta)$ as follows. Fix 
$x(\alpha,\delta) \in \bar{Q}(\alpha, \delta)$ and define
\begin{alignat*}{1}
\bar{Q}_x(\alpha,\delta) = x(\alpha,\delta) + 
\conv(\bar{Q}(\alpha,\delta)-x(\alpha,\delta), x(\alpha,\delta) - 
\bar{Q}(\alpha,\delta))
\end{alignat*}
These are centrally symmetric polytopes with center of symmetry 
at $x(\alpha,\delta)$. When $\bar{Q}$
is scaled by a factor $4$, it is still contained in
$(1+\epsilon)P$, thus we have $2\odot Q_x(\alpha,\delta) \subseteq 
(1+\epsilon)P$. Thus, the 
union of all $\{\bar{Q}_x(\alpha,\delta)\}$ is a \twocov for $K$ 
using at most $2^m 
(\log_{4/3}(1/\epsilon)+1)^m$ symmetric convex bodies.
\end{proof}

\section*{Acknowledgements}
We thank Friedrich Eisenbrand for suggesting to use coverings to 
boost approximate CVP and for helpful remarks and ideas during the research. 
We would also like to thank Christoph Hunkenschröder and Matthias Schymura for 
helpful remarks on the text, and for our stimulating discussions that
boosted our understanding of the closest vector problem.

Part of MN's research was carried out while he was a member of J\'anos Pach's 
chair of DCG at EPFL, supported by Swiss National Science Foundation Grants 
200020-162884 and 200021-165977. 
MN was supported also by the National Research, Development and Innovation Fund (NRDI) grants K119670 and KKP-133864 as 
well as the Bolyai Scholarship of the Hungarian Academy of Sciences and the New National 
Excellence Programme and the TKP2020-NKA-06 program provided by the NRDI. 

MV was supported by the Swiss  National  Science  Foundation  within  
the project \emph{Lattice Algorithms and Integer Programming} (Nr.  
200021-185030).

\bibliographystyle{alpha}
\bibliography{references}

\end{document}